\documentclass[journal]{IEEEtran}
\makeatletter
\newcommand\semihuge{\@setfontsize\semihuge{22.3}{22}}
\makeatother

\makeatletter
\def\underbracex#1#2{\mathop{\vtop{\m@th\ialign{##\crcr
				$\hfil\displaystyle{#2}\hfil$\crcr
				\noalign{\kern3\p@\nointerlineskip}%
				#1\crcr\noalign{\kern3\p@}}}}\limits}

\def\upbracefilla{$\m@th \setbox\z@\hbox{$\braceld$}%
	\bracelu\leaders\vrule \@height\ht\z@ \@depth\z@\hfill 
	\kern\p@\vrule \@width\p@\kern\p@\vrule \@width\p@\kern\p@\vrule \@width\p@
	$}

\def\upbracefillb{$\m@th \setbox\z@\hbox{$\braceld$}%
	\vrule \@width\p@\kern\p@\vrule \@width\p@\kern\p@\vrule \@width\p@\kern\p@
	\leaders\vrule \@height\ht\z@ \@depth\z@\hfill\bracerd
	\braceld\leaders\vrule \@height\ht\z@ \@depth\z@\hfill
	\kern\p@\vrule \@width\p@\kern\p@\vrule \@width\p@\kern\p@\vrule \@width\p@
	$}

\def\upbracefillc{$\m@th \setbox\z@\hbox{$\braceld$}%
	\vrule \@width\p@\kern\p@\vrule \@width\p@\kern\p@\vrule \@width\p@\kern\p@
	\leaders\vrule \@height\ht\z@ \@depth\z@\hfill
	\kern\p@\vrule \@width\p@\kern\p@\vrule \@width\p@\kern\p@\vrule \@width\p@
	$}

\def\upbracefilld{$\m@th \setbox\z@\hbox{$\braceld$}%
	\vrule \@width\p@\kern\p@\vrule \@width\p@\kern\p@\vrule \@width\p@\kern\p@
	\leaders\vrule \@height\ht\z@ \@depth\z@\hfill\braceru$}

\def\upbracefillbd{$\m@th \setbox\z@\hbox{$\braceld$}%
	\vrule \@width\p@\kern\p@\vrule \@width\p@\kern\p@\vrule \@width\p@\kern\p@
	\bracerd\braceld
	\leaders\vrule \@height\ht\z@ \@depth\z@\hfill\braceru$}

\makeatother

\usepackage{tabularx}

\usepackage{algpseudocode}
\usepackage{algorithm}
\usepackage{algorithmicx}

\usepackage{lipsum} 
\usepackage{arydshln}
\usepackage[dvips]{color}
\usepackage{comment}
\usepackage{todonotes}
\usepackage{epsf}
\usepackage{epsfig}
\usepackage{times}
\usepackage{epsfig}
\usepackage{graphicx}
\usepackage{bbold}
\usepackage{bbm}
\usepackage{mathtools}
\usepackage{mathrsfs}
\usepackage{amssymb}
\usepackage{pdfpages}
\usepackage{epstopdf}
\newfloat{algorithm}{t}{lop}

\usepackage{amsmath,amssymb}
\usepackage[makeroom]{cancel}
\usepackage{tikz}

\usepackage{dsfont}
\usepackage{lettrine} 

\usepackage{epsfig,algorithm,algpseudocode,amsthm,cite,url}
\newtagform{blue}{\color{blue}(}{)}
\usepackage{subcaption}
\allowdisplaybreaks
\usepackage{csquotes}
\usepackage{etoolbox}

\usepackage{verbatim}
\usepackage[english]{babel}
\usepackage{diagbox}

\usepackage[justification=centering]{caption}

\usepackage[cmintegrals]{newtxmath}
\usepackage{verbatim}

\newtheorem{theorem}{\bf Theorem}

\newtheorem{definition}{\bf Definition}

\IEEEoverridecommandlockouts

\begin{document}
	%
	\title{\semihuge{Interdependence-Aware Game-Theoretic Framework for Secure Intelligent Transportation Systems}}

	\author{\IEEEauthorblockN{\normalsize Aidin Ferdowsi, \emph{Student Member, IEEE}, Abdelrahman Eldosouky, \emph{Student Member, IEEE}, and	Walid Saad, \emph{Fellow, IEEE}\vspace{-8mm}}
		
	\thanks{Authors are with Wireless@VT, Bradley Department of Electrical and Computer Engineering, Virginia Tech, Blacksburg, VA, USA, {\tt\small \{aidin, iv727, walids\}@vt.edu}. The first two authors contributed equally to this work.}}
\maketitle


	%
	\IEEEpeerreviewmaketitle
	
\begin{abstract}	
	
	The operation of future intelligent transportation systems (ITSs), communications infrastructure (CI), and power grids (PGs) will be highly interdependent.
	In particular, autonomous connected vehicles require CI resources to operate, and, thus, communication failures can result in non-optimality in the ITS flow in terms of traffic jams and fuel consumption. Similarly, CI components, e.g., base stations (BSs) can be impacted by failures in the electric grid that is powering them. Thus, malicious attacks on the PG can lead to failures in both the CI and the ITSs. To this end, in this paper, the security of an ITS against indirect attacks carried out through the PG is studied in an interdependent PG-CI-ITS scenario. In the considered scenario, an attacker can induce non-optimality in the ITS or disrupt a particular set of streets by attacking the PG components while remaining stealthy from the ITS administrators. To defend against such attacks, the administrator of the interdependent critical infrastructure can allocate backup power sources (BPSs) at every BS to compensate for the power loss caused by the attacker. However, due to budget limitations, the administrator must consider the importance of each BS in light of the PG risk of failure, while allocating the BPSs. In this regard, a rigorous analytical framework is proposed to model the interdependencies between the ITS, CI, and PG. Next, a one-to-one relationship between the PG components and ITS streets is derived in order to capture the effect of the PG components' failure on the optimality of the traffic flow in the streets. Moreover, the problem of BPS allocation is formulated using a Stackelberg game framework and the Stackelberg equilibrium (SE) of the game is characterized. Simulation results show that the derived SE outperforms any other BPS allocation strategy and can be scalable in linear time with respect to the size of the interdependent infrastructure.

\end{abstract}
\section{Introduction}

Intelligent transportation systems (ITSs) are complex systems that integrate connectivity, sensing, and autonomy to improve the efficiency and the security of traditional transportation systems \cite{FerdowsiITS}. Different devices and sensors are connected in an ITS to collect, share, and process data of the vehicles and their surroundings. This information exchange requires a reliable communication infrastructure to connect the various vehicles and to transfer the data in real time\cite{saad2019vision}. With its ability to connect ubiquitous devices to the Internet, the Internet of Things (IoT) is seen as a major enabler for future ITSs \cite{Zanella,Chen2014,mozaffari2018beyond,zeng2019joint}. 

As one of the main ITS components, autonomous connected vehicles (ACVs) will receive control signals from the communication infrastructure (CI)'s components such as the base stations (BSs) through vehicle-to-everything (V2X) links. These control signals can help to optimize the operation of the ITS in terms of flow, fuel consumption, and air pollution\cite{Kargl}. However, this reliance on wireless connectivity brings forward new vulnerabilities because of the interdependence that exists between the CI and the ITS\cite{Kargl}. For instance, an ITS attacker can utilize the CI to jam road segments\cite{alpcan2010security}, deny V2X signal services \cite{Lyamin}, tamper with traffic signals \cite{Laszka}, and take the control of ACVs\cite{ferdowsicps}.

Furthermore, since the CI is operated using power grids (PGs), another interdependence relation comes into existence between both infrastructure.
Consequently, ITSs and PGs will also be interdependent through the common CI.
In this regard, failures in a PG will directly affect the CI as well as having indirect effects on the ITS.
Thus, the administrator(s) of such critical infrastructure systems must account for the CI-PG-ITS interdependence, when securing their ITSs. This interdependence exposes the ITS to a new set of attacks such as bad data injection \cite{liu2011false} or physical attacks such as tampering the PG components\cite{he2016cyber}. 

CI-PG-ITS interdependence adds new constraints to the security designs of the ITSs against cyber-physical attacks. For example, a disruption of power delivery from the PG to the CI can deactivate the BSs which send control packets to ACVs. This, in turn, will disrupt the traffic flow in the ITS. Therefore, compensating the power loss at BSs by, e.g., using backup power sources (BPSs) \cite{Wang2019}, needs to consider the impact of the BSs on the ITS traffic flow. However, as the available budgets are usually limited, the administrator of an ITS will need to prioritize the PG components, in the security design, in light of the interdependencies between the ITS, the CI, and the PG. However, this type of interdependencies makes designing security solutions for such systems, highly challenging. 

\subsection{Related Works}

 Critical Infrastructure protection and security has recently attracted significant attention \cite{jamei2016micro} and \cite{eldosouky2015contract}. In general, critical infrastructure systems refer to the systems that are vital to modern day economies and cities. Examples of such systems include power grids, transportation systems, nuclear reactors, communications infrastructure, water
supply, and financial services \cite{keeney2005insider}. Therefore, securing and maintaining the proper operation of such systems is of utmost priority. However, one challenge to the security of critical infrastructure systems stems from their interdependent nature, i.e., the functionality of one infrastructure can depend on one or more other infrastructure. Therefore, the failure of one infrastructure can affect other dependent infrastructure systems.

The security of \emph{interdependent critical infrastructure (ICI)} has thus been the focus of many recent works \cite{Rahnamay,Parandehgheibi2014,Das2014,Chen2018,Ferdowsi2017}. For instance, the authors in \cite{Rahnamay} solved a power flow optimization problem for interdependent PG-CI that takes into account the power requirements of the CI and the impact of the CI on the PG state estimation. The work in \cite{Parandehgheibi2014} developed a power load control policy for PGs that mitigates the cascading failures of interdependent PG-CI while the authors in \cite{Das2014} analyzed the root cause of failures in ICIs. Furthermore, the work in \cite{Chen2018} studied robustness of large-scale interdependent CI-PG via a complex network analysis. In addition, the authors in \cite{Ferdowsi2017} and \cite{FerdowsiCBG} proposed a game-theoretic security solution against data injection attacks on CI components that impact PGs.

Moreover, the security of connected vehicles has gained a lot of attention recently due to the important role of ACVs in the ITSs \cite{FerdowsiITSC,Lei,Chetlur,Zheng2017}. The authors in \cite{FerdowsiITSC} proposed a deep reinforcement learning algorithm that makes ACVs robust against cyber attacks on the CI. The work in \cite{Lei} introduced a blockchain-based trust management mechanism for interdependent CI-ITSs that takes into account the geographical layout of the ITS networks. Furthermore, the authors in \cite{Chetlur} modeled the interdependencies between a CI and an ITS using a Poisson line process. Meanwhile, in \cite{Zheng2017}, the authors studied the cyber attacks on intersection controllers of an ITS using game theory.

However, the works in \cite{Rahnamay,Parandehgheibi2014,Das2014,Chen2018,Ferdowsi2017,FerdowsiCBG,FerdowsiITSC,Lei,Chetlur,Zheng2017} did not consider stealthy attacks on the ICI in which the attacker aims at disturbing the ICI while staying stealthy from detection. Furthermore, the works in \cite{Rahnamay,Parandehgheibi2014,Das2014,Chen2018,Ferdowsi2017,FerdowsiCBG} only studied PG-CI interdependencies while \cite{FerdowsiITSC,Lei,Chetlur,Zheng2017} focused solely on the CI-ITS interdependencies and there has been no study on the indirect but pronounced effects of PG failures on ITSs. Although the work in \cite{ferdowsi2017colonel} considers indirect dependencies of critical infrastructure such as gas and water networks on a CI through a PG, however, the attack model therein did not consider any stealthiness for the attacker.

\subsection{Contributions}
The main contribution of this paper is, thus, a holistic game-theoretic framework that analyzes the security of interdependent PG-CI-ITS infrastructure system.
The proposed framework addresses the security of ITSs against indirect attacks carried out through the PGs as these attacks have direct effects on the CI and indirect effects on the ITSs.
In particular we have the following key contributions:
\begin{itemize}
	\item We develop a novel model for capturing the interdependencies between PGs, CI, and ITSs. In particular, we analytically derived the interdependence relations across the three infrastructure: PG, CI, and ITS through formulating a two-tier model for the ITS-CI interdependence as well as the CI-PG interdependence.
	\item Combining these two interdependence models, we derive the full PG-CI-ITS interdependence as a one-to-one mapping that captures the effect of PG components on the ITS operation.
	These one-to-one relations can be used by ICI administrators, when securing their interdependent infrastructure, to prioritize the PG components based on their ultimate effect on the ITS operation.
	\item We model the interactions between a stealthy attacker and the administrator of an interdependent PG-CI-ITS system, using game theory. In particular, we formulate a Stackelberg game to model such interactions in which the attacker acts as a follower whose goal is to disrupt the ITS flow through attacking the PG components. The ICI administrator acts as a leader whose goal is to minimize such disruption by maintaining the optimal operation of its CI. In this game, the defender uses the interdependence model to allocate the BPSs to its CI to maintain the ITS operation.
	\item We analytically derive the Stackelberg equilibrium (SE), for the proposed game, which is used to characterize the optimal attack and defense mechanisms. We show that the derived SE strategy, for the administrator, is scalable in linear time, and, thus it is practical for large-scale ICI implementations.
	\item Through simulations, we show that the proposed SE strategy can outperform any other security strategy for protecting ICIs.
\end{itemize}

The rest of the paper is organized as follows. Individual infrastructure models and the interdependence models are presented in Section \ref{sec:model}. The attacker’s stealthy model and the optimal defense strategy are derived in Section \ref{sec:attackdefense}. The proposed Stackelberg game between the attacker and the ICI administrator is formulated in Section \ref{sec:game} where the equilibrium solutions are also derived. Simulation results are shown in Section \ref{sec:sim}. Finally, conclusions are drawn in Section \ref{sec:conc}.

\section{System Model}\label{sec:model}

\subsection{Individual Infrastructure Models}

\subsubsection{ITS Model}
Consider an ITS that is modeled by a set $\mathcal{N}$ of $N$ intersections. This ITS has three main macroscopic characteristics at each street $ ij $ (direction of movement is from intersection $ i $ to intersection $ j $) \cite{daganzo1997fundamentals}:
\begin{itemize}
	\item \emph{Flow}, $ q_{ij}(t) $, which is the number of ACVs passing street $ ij $ over a given period of time (expressed in veh/h/lane)
	\item \emph{Density}, $ k_{ij}(t) $, which is the number of ACVs moving in street $ ij $ at a specific instant in time (expressed in veh/km/lane)
	\item \emph{Space-mean-speed}, $ v_{ij}(t) $, which is the average rate of motion for vehicles moving in street $ ij $ (expressed in km/h).
\end{itemize}

In an optimal ITS which takes into account minimum travel time, maximum safety, and minimum air pollution, every street $ ij $ is designed to have an optimal flow, $ \bar{q}_{ij} $. Moreover, let $\mathcal{O}_i$ and $\mathcal{I}_i$ be the set of streets that have flow from and to intersection $i$. Then, at every intersection $i$, we have:
\begin{align}\label{eq:intersection}
\bar{q}_{ji} = \sum_{k \in \mathcal{O}_i, k \neq j}a_{ik}\bar{q}_{ik}, \forall j \in \mathcal{I}_i,
\end{align}
where $ a_{ik} $ is the portion of flow $ \bar{q}_{ik} $ that comes from $ \bar{q}_{ji} $.

In fact, \eqref{eq:intersection} captures the fact that the inflow from every street to an intersection is divided into outflows from that intersection. Note that we ignore u-turns at intersections since typically u-turns are a small fraction of the through traffic. We define an $ n\times n $ network flow matrix, $ \boldsymbol{Q} $, such that the element at row $i$ and column $j$ is $a_{ij} $. Therefore, in order to find the optimal values for the traffic flow at every street we need to solve:
\begin{align}\label{eq:ITSflows}
	\boldsymbol{I} \boldsymbol{q} = \boldsymbol{Q} \boldsymbol{q},
\end{align}
where $ \boldsymbol{I} $ is an $ n\times n $ identity matrix and $ \boldsymbol{q} $ is an $ n\times 1 $ vector that contains the flow rate values of all of the streets. Equation \eqref{eq:ITSflows} can be written as $(\boldsymbol{I}-\boldsymbol{Q})\boldsymbol{q} = \boldsymbol{0}$ or $\boldsymbol{A}\boldsymbol{q} = \boldsymbol{0}$, where $ \boldsymbol{A} \triangleq \boldsymbol{I}-\boldsymbol{Q} $. $ \boldsymbol{A} $ can be proven to be under-determined since its rank is $ n-1 $\cite{Harrod1984,Zhou}. Therefore, in order to solve \eqref{eq:ITSflows}, we will assume to know the traffic flow value at least in one street and, then, we can define an $ n\times n-1 $ matrix $A_i$ that is identical to $ \boldsymbol{A}$ with the $ i $-th column removed. We also define $ \boldsymbol{a}_i $ as the $ i $-th column of matrix $ \boldsymbol{A} $.

Next, the values of the flow in the remaining streets can be calculated by solving the equation:
\begin{align}\label{eq:reducedflow}
	\boldsymbol{A}_i \boldsymbol{q}_{\hat{i}} = -\boldsymbol{a}_i q_i,
\end{align}
where $ \boldsymbol{q}_{\hat{i}} $ is an $n-1 \times 1$ vector containing all the values of street flows other than street $ i $ and $ q_i $ is the value of the known street $ i $'s traffic flow. The solution of \eqref{eq:reducedflow} can be found by\cite{Harrod1984}:
\begin{align}\label{eq:reducedflowsol}
	\boldsymbol{q}_{\hat{i}} = -  (\boldsymbol{A}_i^{T}\boldsymbol{A}_i)^{-1}\boldsymbol{A}_i^{T}\boldsymbol{a}_i q_i.
\end{align}

To this end, we can find the effect of flow deviation at the $ i $-th street on the entire ITS using \eqref{eq:reducedflowsol} . Let $ \boldsymbol{\delta}_i $ be the flow deviation at street $ i $ and $ \tilde{\boldsymbol{q}}_{\hat{i}} $ be the flow rate of other streets following a deviation at street $i$, then, we have:
\begin{align}
	\tilde{\boldsymbol{q}}_{\hat{i}} &= -  (\boldsymbol{A}_i^{T}\boldsymbol{A}_i)^{-1}\boldsymbol{A}_i^{T}\boldsymbol{a}_i (q_i - \delta_i)\nonumber\\
	&= -  (\boldsymbol{A}_i^{T}\boldsymbol{A}_i)^{-1}\boldsymbol{A}_i^{T}\boldsymbol{a}_i q_i + (\boldsymbol{A}_i^{T}\boldsymbol{A}_i)^{-1}\boldsymbol{A}_i^{T}\boldsymbol{a}_i\delta_i\nonumber\\
	& = \boldsymbol{q}_{\hat{i}} + (\boldsymbol{A}_i^{T}\boldsymbol{A}_i)^{-1}\boldsymbol{A}_i^{T}\boldsymbol{a}_i\delta_i.\label{eq:deviation}
\end{align}

Thus, the flow deviation $ \boldsymbol{\delta}_{\hat{i}} $ on streets other than street $ i $, can be derived by solving $ \boldsymbol{\delta}_{\hat{i}} \triangleq \tilde{\boldsymbol{q}}_{\hat{i}} - \boldsymbol{q}_{\hat{i}} =  (\boldsymbol{A}_i^{T}\boldsymbol{A}_i)^{-1}\boldsymbol{A}_i^{T}\boldsymbol{a}_i\delta_i.$ 

This result represents the first step in modeling the interdependence in PG-CI-ITS as it allows an administrator to fully understand the non-optimality in the flow of its entire ITS system. Next, we study the model of the CI in light of its connection to the ITS.

\subsubsection{CI Model}
In our model, the ACVs of the ITS are powered by a CI that consists of a set $\mathcal{B}$ of $B$ base stations. Each BS covers a portion of streets and communicates important control messages to the ACVs within the covered areas. As is customary in cellular networks, we use a hexagonal shape to model the coverage area of every BS as shown in Fig. \ref{fig:model}. In Fig. \ref{fig:model}, we can see that every BS can cover the entirety of a street or a section of every street. Under normal operating conditions, every BS is expected to service all the ACVs in its coverage region.

Next, we study the operation of a BS, and its impact on the ITS, in case there is a disruption in the delivered electrical power.
We define $ p_b^o $ as the power required by every BS $ b $, in order to be activated \cite{conte2012power}. Similarly, we define $ p_b^t $ as the power required by a BS $ b $ to send control packets to $100\%$ of the ACVs in its cell, as shown in Fig. \ref{fig:bspower}. Clearly, if the BS received only $ p_b^o $, it will not be able to serve any ACVs.

Then, consider that BS $ b $ received a power $ p_b^r $ such that $ p_b^o\leq p_b^r \leq p_b^t $, then, BS $ b $ will be able to send packets to only a fraction of the ACVs in its cell. From Fig. \ref{fig:bspower}, this fraction of users can be given as:
\begin{equation}
x_b  = \frac{p_b^r -  p_b^o }{p_b^t - p_b^o},
\end{equation} 
where $x_b$ is the fraction of users that can be covered by BS $ b $.

Next, let $ p_b^a $ be the power deviation at BS $ b $ such that $  p_b^a = p_b^t - p_b^r $. Then the relationship between the percentage of users that receive packets from BS $ b $ and the power deviation $ p_b^a $ can be shown as in Fig. \ref{fig:powervsuser}. From Fig. \ref{fig:powervsuser}, we can see that, if the supplied power deviates by more than $ p_b^t - p_b^0 $, then, the BS cannot send packets to any of the users.

\begin{figure}[!t]
	\centering
	\includegraphics[width=\columnwidth]{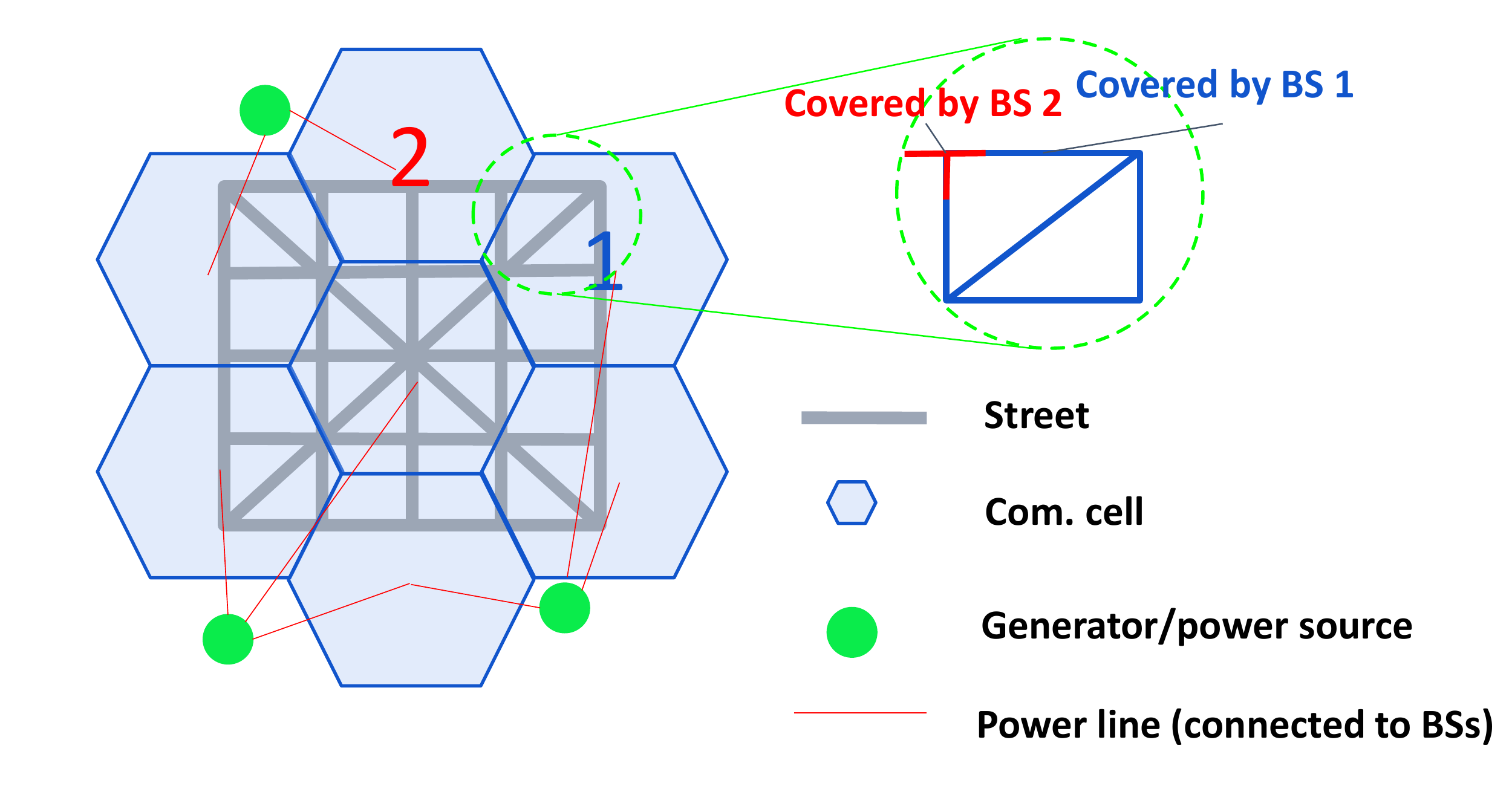}
	\caption{An illustrative example of an ITS and a CI.}
	\label{fig:model}
\end{figure}

\begin{figure}[!t]
	\centering
	\begin{tikzpicture}
		\draw[thick,->] (0,0) -- (3.5,0) node[anchor=north west] at(0.5,-0.35) {Cell load (\%)}
		node at (3,-0.25) {$100$}
		node at (0,-0.25) {$0$} ;
		\draw[thick,->] (0,0) -- (0,3.5)
		node at (-0.25,1.5) {$p_b^o$} 
		node[anchor=south east,rotate = 90] at (-0.5,3.5) {Power consumption (W)};
		\filldraw[fill=blue!40!white, draw=black] (0,1.5) -- (3,3) -- (3,0) -- (0,0);
		\draw[dashed] (0,3) -- (3,3)
		node at (-0.25,3) {$p_b^t$} ;
	\end{tikzpicture}
	\caption{Power consumption diagram of every BS $ b $.}
	\label{fig:bspower}
\end{figure}
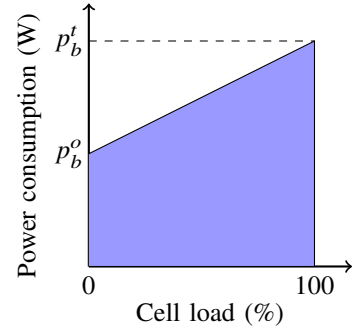

\begin{figure}[!t]
	\centering
	\begin{tikzpicture}
	\draw[thick,->] (0,0) -- (3.5,0) node[anchor=north west] at(0.5,-0.35) {Power deviation (W)}
	node at (1.5,-0.25) {$p_b^t - p_b^o$} ;
	\draw[thick,->] (0,0) -- (0,3.5)
	node[anchor=south east,rotate = 90] at (-0.5,3.5) {User coverage (\%) };
	\filldraw[fill=red!60!white, draw=black] (0,3) -- (1.5,0) -- (0,0)
	node at (-0.3,3) {$100$}	
	node at (-0.3,0) {$0$} ;
	\end{tikzpicture}
	\caption{Percentage of user coverage inside a cell as a function of deviation on the BS received power.}
	\label{fig:powervsuser}
\end{figure}
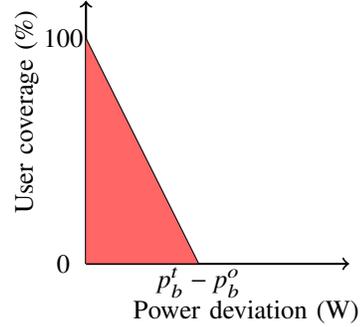

To this end, the previous analysis of the CI allows us to determine the percentage of the affected ACVs in case of power disruption. Next, we study the the operation of the PG in light of its connections to the CI.

\subsubsection{PG Model}
The BSs in our CI are powered by a PG. In particular, we model this PG by a graph $\mathcal{P}=\{ \mathcal{V},\mathcal{E}\}$ in which $\mathcal{V}$ is the set of power nodes and $\mathcal{E}$ is the set of connection lines between these nodes. We consider two types of power nodes: Power generators and loads. Note that, in a typical power grid, there can be other types of nodes such as transmission units and substations. However, we do not model these units explicitly as nodes in the graph $\mathcal{G}$ as they do not have direct effect on the communication network.

Since we are interested in the power received at the BSs, we consider them in more detail as being load nodes in the power grid.
The effect of the BSs failure can be included as part of the load nodes failure. Thus, let $\mathcal{G}$ be the set of $G$ generators, $\mathcal{L}$ be the set of all load nodes, and $\mathcal{Q}$ be the set of $Q$ non-BS loads. Thus, we have $\mathcal{L} = \mathcal{Q} \cup \mathcal{B}$, where $\mathcal{B}$ is the set of base stations as defined earlier. We also have $\mathcal{V} = \mathcal{G} \cup \mathcal{L}$ as the set of all nodes include all the generators and all the load nodes.

Since we are interested in studying the dependency between the PG and the CI, it is important to highlight the effect of power generation failures on CI. In a typical PG, electricity is generated to match the power consumption which is known as demand-response~\cite{palensky2011demand}. As it is hard and inefficient to store electricity, the power grid uses a means of communication to organize the generation capacity of each power generation unit. This management is also useful in case one generator fails so that its planned load can be shifted to other generation units to meet the power consumption demand~\cite{eldosouky2017resilient}. The relation between the power grid failures and the CI, is studied next.

\begin{figure*}
	\includegraphics[width=\textwidth]{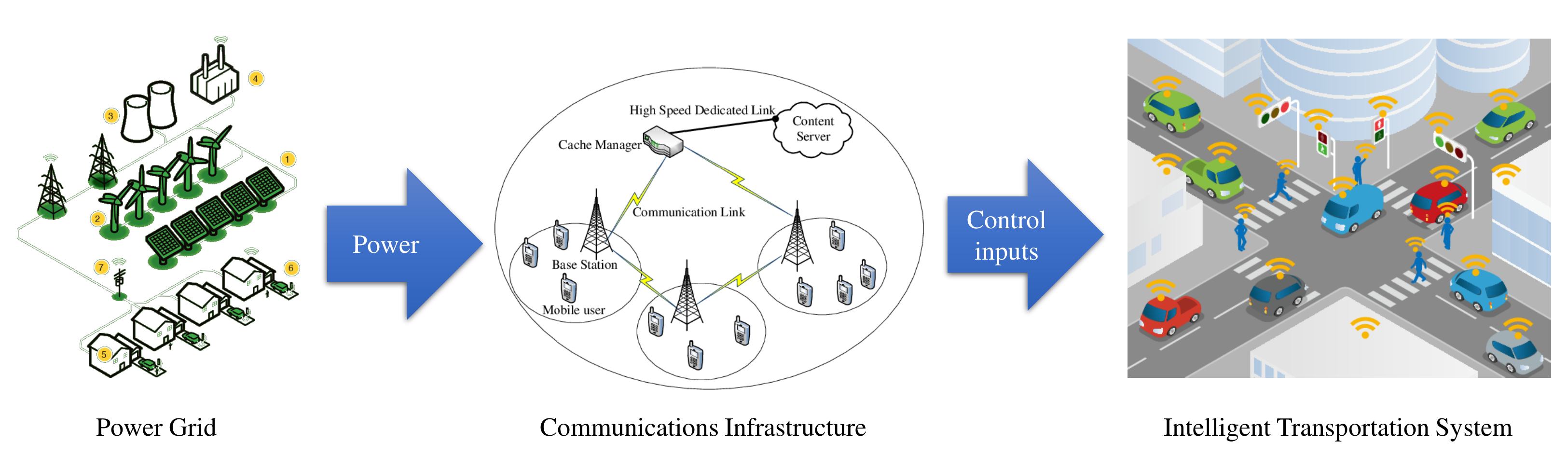}
	\caption{An illustration of the interdependencies between ITS, CI, and PG.}
	\label{Fig:inter}
\end{figure*}

\subsection{Interdependence Analysis}

In Fig. \ref{Fig:inter}, we illustrate the interdependencies in our system. From Fig. \ref{Fig:inter} we can see that the power grid provides the essential electricity required for the proper operation of the BSs in the CI. The BSs in turn use this power to transmit the control signals to the ITS. Therefore, in case a failure occurs in any of the power lines, some BSs might not receive their required power. As a result, they will fail to cover some portions of the ITS. Consequently, some ACVs in the ITS will not receive the control packets from the CI which will lead to a non-optimal traffic flow in the ITS. 

In such a scenario, it can be said that the operation of the CI is directly dependent on the functionality of the PG while the operation of the ITS is indirectly dependent on the PG, through the CI. As a result, next, we propose a two-tier model to capture the interdependence within a PG-CI-ITS system. In particular, we will use the empirical, agent-based and flow-based methods to model the interdependencies between the ITS and CI in the first tier. Then, we will use the a network-based model to study the interdependence between the CI and the PG \cite{OUYANG201443}, in the second tier. 

\subsubsection{Interdependence of ITS and CI}
In our ITS model, every street or every section of each street is covered by a specific BS. This makes the ITS vulnerable against possible failures in the CI whenever a portion of BSs are either deactivated or do not have enough power to send control packets to all of the ACVs in streets. In this case, the flow at every street might deviate from its optimal value. In particular, if a BS $ b $ can only send control packets to a fraction $x$ of its covered street $ i $, then the flow reduction at street $ i $ will be $ -x \frac{s_{ib}}{l_i} \delta $, where $s_{ib}$ is the length of street $ i $ covered by BS $ b $, $ l_i $ is the length of street $ i $, and $ \delta $ is a fixed number indicating the difference between the flow of a street that is $1$ km long before and after deviation.

Such flow deviation can propagate forward and backward as time goes and can essentially affect the entire ITS. Thus, we define a matrix $\boldsymbol{C}$ in which each element at row $i$ and column $b$, $c_{ib}$ equals:
\[
c_{ib} = \begin{cases}
\frac{s_{ib} }{l_i} & \text{if BS $ b $ covers $s_{ib}$ of street $i$,}  \\
0 & \text{otherwise.}
\end{cases}
\]

Then, if the received power at a BS $ b $ is reduced by $ x $, such that $ p_b^o\leq x \leq p_b^t $, then its effect on the flow of street $ i $ is $ -\frac{x}{p_b^t - p_b^o} \frac{s_{ib}}{l_i} \delta $, and, thus, the flow deviation will propagate to the entire ITS which can be derived using \eqref{eq:deviation} by:
\begin{equation} 
\boldsymbol{\delta}_{\hat{i}}^b =  (\boldsymbol{A}_i^{T}\boldsymbol{A}_i)^{-1}\boldsymbol{A}_i^{T}\boldsymbol{a}_i\frac{x}{p_b^t - p_b^o} \frac{s_{ib}}{l_i}\delta.
\end{equation}

Then, if we concatenate -1 at the $ i $-th row of $(\boldsymbol{A}_i^{T}\boldsymbol{A}_i)^{-1}\boldsymbol{A}_i^{T}\boldsymbol{a}_i$ and shift all the rows after row $ i $ one row down, we will construct a new vector $ \boldsymbol{e}_{i}^b $ that has $ n $ rows and shows the effect of failure of BS $ b $ on the ITS through direct impact on street $i$.

Therefore, the total impact of a power deviation $ x $ at BS $ b $ on ITS can be written as follows:
\begin{align}\label{eq:impact}
	\boldsymbol{e}^b(x) = \underbrace{\left(\sum_{i \in \mathcal{S}_b}  \boldsymbol{e}_{i}^b \frac{s_{ib}}{l_i(p_b^t - p_b^o)}\right) }_{\boldsymbol{z}^b}\delta x = \boldsymbol{z}^b\delta x,\,\, \forall x \in \left[0,p_b^t-p_b^o\right],
\end{align}
where set $ \mathcal{S}_b $ is the set of all the streets covered by BS $ b $.
In \eqref{eq:impact}, the $ i $-th element $ z_i^b $ of $ \boldsymbol{z}^b $ captures the \emph{one-to-one impact of BS $ b $ on the street $ i $}.

Note that, in order to compare the value of every BS with that of other BSs, we can use the $ \mathcal{L}^1 $-norm of $ \boldsymbol{z}^b $ which we define it as $ z_b \triangleq \left|\boldsymbol{z}_b\right| $, such that a BS with higher $ z_b $ value will have higher impact on the ITS.

\subsubsection{Interdependence of CI and PG}
The BSs must be connected to the PG to obtain the electricity required for their operation. Thus, we model each BS as a load node in the power grid. As discussed earlier, these BSs need certain power requirements in order for them to operate properly otherwise they will fail\cite{conte2012power}. These power requirements can be satisfied by the power generators that are connected to the BS using transmission lines.

To study the effect of a power generator failure on its connected BSs, we consider the example in Fig. \ref{fig:model}. In Fig. \ref{fig:model}, power sources are represented using circles while the BSs are represented using hexagons. We notice that some BSs can receive electricity directly from multiple generators or indirectly through the transmission lines from other generators. Similarly, each generator is directly connected to some BSs (as loads) and also to the rest of the grid using transmission lines.

Modeling the exact behavior of power generators failure is a complex process as it can involve multiple failures in the grid known as cascading failures~\cite{Parandehgheibi2014}. However, in this work, we are concerned only with the effect on the BSs. We evaluate each power generation based on its failure effect on the connected BSs. The failure here refers to the inability of the power generator to produce the electricity either fully or partially due to any disruptive events such as cyber or physical attacks. In the following, we explain the procedure of evaluating the power generators in case of full failure, i.e., no electricity generation. Partial failures can be modeled in a similar way by considering the affected BSs. 

Let $ \boldsymbol{T} $ be a $ B \times G $ matrix such that the entity at row $ b $ and column $ i $ of $ \boldsymbol{T} $ which we define as $ t^g_b $ is a value in $ [0,1] $ that indicates the portion of $ p_b^r $ received at BS $ b $ from power source $ g $. Therefore, if a power source $ g $ is not connected to BS $ b $, then $ t^g_b = 0 $. When a BS is connected to more than one generator, we will have $ \sum_{g \in \mathcal{G}}t^g_b = 1, \,\, \forall b \in \mathcal{B}$, i.e., the summation over all the connected generators will equal the full power received at a single BS, from these generators. Next, we explain the procedure that an attacker can use to exploit the interdependence in order to perform its stealthy attacks.

\section{ Interdependent PG-CI-ITS Under Attack}\label{sec:attackdefense}
\subsection{Stealthy Attack Model}\label{sec:stealthyattack}
Consider an adversary who aims at disrupting the operation of the ITS by attacking the PG components. Such an attack can target either the generators or the power lines. When the attacker damages a generator, that generator will no longer supply power at full capacity. In addition, the attacker can specifically reduce the power supply at any line by damaging the infrastructure. Doing so, the goal of the attacker can include disrupting the flow of a specific street or the entire ITS. However, defining such a specific goal for the attacker might not be always possible since the attacker may not have access to all the PG components.

Although a higher disruption at the PG will cause a higher flow deviation at the ITS, it will on the other hand expose the attacker to be detected with higher chance. Therefore, we introduce the \emph{stealthiness level} to represent the level at which the attacker risks to be detected, i.e., a risk averse attacker will adopt a higher stealthiness level, and, thus it will perform less attacks in order not to be detected. On the other hand, a risk tolerant attacker will adopt a lower stealthiness level by performing a large scale attack despite the higher chance of being detected. Here, we define the stealthiness level based on the location of the attack as each location can cause a different degree of damage to the system while having a different chance of detection. In particular, we define three levels of stealthiness: a) at the power source level, b) at the power line level, and c) at the BS level.

Let $ p^a_{gb} $ be the power deviation caused by the attacker at the line connecting power source $ g $ to BS $ b $. Then, at a power source $ g $ level, the probability of being detected can be defined as:
\begin{align}\label{eq:stealthygen}
	\pi_g = \frac{\sum_{b \in \mathcal{B}}p^a_{gb}}{\sum_{b \in \mathcal{B}}t^g_b p^t_b},
\end{align}
where the numerator is the total power deviation by the attacker at power source $ g $, and the denominator is the total generated power by the power source $ g $ in a safe scenario.

Similarly, the probability of attack detection at the power line level connecting a power source $ g $ to a BS $ b $ can be defined as:
\begin{align}\label{eq:stealthyline}
	\pi_{gb} = \frac{p^a_{gb}}{t^g_b p^t_b}.
\end{align}

Moreover, the probability of staying stealthy at a BS $ b $ can be defined as:
\begin{align}\label{eq:stealthybs}
	\pi_{b} = \frac{\sum_{g \in \mathcal{G}}p^a_{gb}}{p^t_b - p^o_b},
\end{align}
where the numerator is the total power deviation at BS $ b $ and the denominator is the required power for covering $100\%$ of the users in BS $ b $'s cell.

The probabilities of detection defined in \eqref{eq:stealthygen}, \eqref{eq:stealthyline}, or \eqref{eq:stealthybs} will be utilized in Section \ref{sec:game} to study the attacker's behavior when interacting with a defender adopting our defense strategy discussed next.

\subsection{Defense Strategy}
As a countermeasure, the administrator of the ICI, the \emph{defender} hereinafter, can allocate BPSs at every BS in order to compensate for the power loss at the BS. However, in practice, due to the budget limitation and different impact levels of each BS on the ITS, the administrator must allocate a different amount of BPSs at every BS. Let $ P^d $ be the total available amount of BPSs\footnote{Since the BPSs can be designed to have any desired storage capacity we consider $ P^d $ and $ p^d_b $ to take any positive value.} for the defender and $ p^d_b $ be the allocated BPSs at BS $ b $, then we will have $ \sum_{b \in \mathcal{B}} p^d_b \leq P^d $, i.e., the defender does need to allocate all the BPSs. Therefore, the total power deviation at BS $ b $ can be written as $ \sum_{g \in \mathcal{G}}p^a_{gb} - p_b^d $.

The defender can then evaluate the outcome from allocating each BPS by evaluating the improvement in the ITS due to the compensated power from the BPS. However, the defender's outcome from allocating a BPS will depend on the attacker's choice of PG component. Thus, next we study the outcomes for both players in presence of these interactions.

\subsection{Attacker-Defender Interactions}
We propose to define the defender's payoff from allocating the BPSs at BSs, as a function in both players' actions, as the negative of total flow deviation at the ITS using \eqref{eq:impact} as follows:
\begin{align}\label{eq:defenderpayoff}
	u^d(\boldsymbol{p}^d,\boldsymbol{p}^a) = -\sum_{b \in \mathcal{B}}z_b \left(\sum_{g \in \mathcal{G}}p^a_{gb} - p_b^d\right), 
\end{align}
where $ \boldsymbol{p}^d $ and $ \boldsymbol{p}^a $ are the strategy vectors of the defender and the attacker, respectively. Here, the defender's strategy $ \boldsymbol{p}^d $ represents the allocated BPSs at every BS, while the attacker's strategy $ \boldsymbol{p}^a $ represents the targeted power deviations at every power line. Recall that $z_b $ is given by \eqref{eq:impact}, and it represents the one-to-one impact of the BSs on the ITS. We note that, in \eqref{eq:defenderpayoff}, we dropped $ \delta $ from \eqref{eq:impact} as it is a constant value and will not impact the strategy design. 

We can rearrange \eqref{eq:defenderpayoff} as follows:
\begin{align}
u^d(\boldsymbol{p}^d,\boldsymbol{p}^a) = - \left[\sum_{g \in \mathcal{G}}\sum_{b \in \mathcal{B}}z_bp^a_{gb} - \sum_{b \in \mathcal{B}}z_b p_b^d\right]. 
\end{align} 

Since the defender and the attacker have opposing goals, the defender's loss is considered as the attacker's gain. Therefore, the attacker's payoff will be the negative of the defender's payoff, as follows:
\begin{align}\label{eq:attpayoff}
u^a(\boldsymbol{p}^d,\boldsymbol{p}^a) =  \sum_{g \in \mathcal{G}}\sum_{b \in \mathcal{B}}z_bp^a_{gb} - \sum_{b \in \mathcal{B}}z_b p_b^d. 
\end{align} 

However, \eqref{eq:attpayoff} represents the general outcome of the attacker in case the attack is overt and no stealthiness was adopted. When an attacker performs a stealthy attack, its payoff will depend on the level of stealthiness. In this case, its payoff can be given by:
\begin{align}\label{eq:attpayoff_stealthy}
u^a(\boldsymbol{p}^d,\boldsymbol{p}^a) =  \sum_{g \in \mathcal{G}}\sum_{b \in \mathcal{B}}z_bp^a_{gb}(1-\pi_x) - \sum_{b \in \mathcal{B}}z_b p_b^d, x \in \{g,gb,b\},
\end{align} 
where $x$ represents the stealthiness level at the power source, at the power line, or the at BS. Thus, we can define three different payoff functions based on the three levels os stealthiness defined in Section \ref{sec:stealthyattack}. First, for the power source level stealthiness, we substitute \eqref{eq:stealthygen} in \eqref{eq:attpayoff_stealthy}, so we get:

\begin{align}\label{eq:attpayoffgen}
	u^a_\mathcal{G}(\boldsymbol{p}^a,\boldsymbol{p}^d) = \sum_{g \in \mathcal{G}}\sum_{b \in \mathcal{B}}z_bp^a_{gb} \left(1- \frac{\sum_{b \in \mathcal{B}}p^a_{gb}}{\sum_{b \in \mathcal{B}}t^g_b p^t_b}\right) - \sum_{b \in \mathcal{B}} z_bp_b^d. 
\end{align}

From \eqref{eq:attpayoffgen}, we can see that large values for each $ p_{gb}^a $ will yield a larger flow deviation in the ITS. However, this will result in a higher probability of detection. Therefore, the attacker must choose $ \boldsymbol{p}^a $ such that $ \frac{\sum_{b \in \mathcal{B}}p^a_{gb}}{\sum_{b \in \mathcal{B}}t^g_b p^t_b}\leq  1 , \forall g \in \mathcal{G} $.

Similarly, the attacker's payoff while seeking to remain stealthy at every power line level can be calculated by substituting \eqref{eq:stealthyline} in \eqref{eq:attpayoff_stealthy}, so we get:
\begin{align}\label{eq:attpayoffline}
	u^a_{\mathcal{L}}(\boldsymbol{p}^a,\boldsymbol{p}^d) = \sum_{b \in \mathcal{B}}z_b \left(\sum_{g \in \mathcal{G}}p^a_{gb} \left(1-\frac{p^a_{gb}}{t^g_b p^t_b}\right) - p_b^d\right).
\end{align}

In \eqref{eq:attpayoffline}, we can see that, in order to stay stealthy, the attacker must choose $ \boldsymbol{p}^a $ such that $ \frac{p^a_{gb}}{t^g_b p^t_b}\leq  1$.

Finally, for the case in which the attacker wants to stay stealthy at the BS level, the payoff function can be calculated by substituting \eqref{eq:stealthybs} in \eqref{eq:attpayoff_stealthy}, so we get:
\begin{align}\label{eq:attpayoffbs}
	u^a_{\mathcal{B}}(\boldsymbol{p}^a,\boldsymbol{p}^d) = \sum_{b \in \mathcal{B}}z_b \left(\sum_{g \in \mathcal{G}}p^a_{gb} - p_b^d\right)\left(1-\frac{\sum_{g \in \mathcal{G}}p^a_{gb}}{p^t_b - p^o_b}\right).
\end{align}

As we can see from \eqref{eq:defenderpayoff}, \eqref{eq:attpayoffgen}, \eqref{eq:attpayoffline}, and \eqref{eq:attpayoffbs}, the payoffs are functions of both of the strategies of the defender and the attacker which motivates a game-theoretic approach \cite{han2012game}. Moreover, since BPSs will first be implemented and then the attacker will attack the PG, thus, the defender must choose its best strategy before seeing the attacker's strategy. This scenario can be properly modeled using Stackelberg games \cite{han2012game} in which the defender is the leader and the attacker is the follower. Note that, this hierarchical model of the players is common to many noncooperative Stackelberg games when addressing security problems e.g., \cite{maharjan2013dependable} and \cite{eldosouky2019drones}.

\section{Stackelberg Game Formulation}\label{sec:game}
We formulate a single-leader, single-follower Stackelberg game \cite{han2012game}, between the defender and the attacker. The defender (leader), will act first by choosing $ \boldsymbol{p}^d $ to maximize its payoff. The attacker, having seen the attacker's allocated BPSs, will engage in a noncooperative game by choosing $ \boldsymbol{p}^a $ to maximize its payoff. In fact, the final flow at the ITS is a function of the defender and the attacker's strategies. One suitable concept to find the optimal strategies, and solve the proposed game, for both the attacker and the defender is that of a Stackelberg equilibrium (SE) as defined next.
\begin{definition}
	A strategy profile $ \left(\boldsymbol{p}^{d^*},\boldsymbol{p}^{a^*}\right) $ is a \emph{Stackelberg equilibrium} if it satisfies the following conditions:
	\begin{align}
		u^a_{x}(\boldsymbol{p}^{a^*},\boldsymbol{p}^{d^*}) &\geq u^a_{x}(\boldsymbol{p}^{a},\boldsymbol{p}^{d^*}),\\
		u^d(\boldsymbol{p}^{d^*},\boldsymbol{p}^{a^*}) & =  \max_{\boldsymbol{p}^{d}} u^d(\boldsymbol{p}^{d},\boldsymbol{p}^{a^*}),
	\end{align}
	where $ x $ can be $ \mathcal{G} $, $ \mathcal{B} $, or $ \mathcal{L} $ depending on the desired level of stealthiness.
\end{definition}

According to this definition, the defender needs to choose a strategy that maximizes its outcome based on the attacker's optimal response. Therefore, in order to find the SE we can proceed by backward induction. First, we need to derive the values of $ \boldsymbol{p}^a $ that maximize the attacker's payoff. Then, we will find the defender's strategy at SE by plugging in the attacker's maximizer strategy into the defender's payoff function and finding the maximizer strategy of the defender. In what follows, we derive the SE for the three stealthiness levels.

\subsection{Stealthiness at the Power Source Level}
In this case, the attacker's problem of finding the values of $ \boldsymbol{p}^a $ that maximize its payoff can be formulated according to \eqref{eq:attpayoffgen} as follows:
\begin{align}\label{eq:maximizerpslevel}
	&\max_{\boldsymbol{p}^a} \sum_{g \in \mathcal{G}}\sum_{b \in \mathcal{B}}z_bp^a_{gb} \left(1- \frac{\sum_{b \in \mathcal{B}}p^a_{gb}}{\sum_{b \in \mathcal{B}}t^g_b p^t_b}\right) - \sum_{b \in \mathcal{B}} p_b^d.\\
	&\text{s.t.} \sum_{b \in \mathcal{B}}p^a_{gb} \leq\sum_{b \in \mathcal{B}}t^g_b p^t_b\,\, \forall g \in \mathcal{G}.\label{cond10}
\end{align}

We can assume that the power loss at the generator is equally distributed at every power line connected to it. Thus, we can find the SE by considering the values for $ p^a_{gb} $ are equal for every $ g \in \mathcal{G}$. The following theorem derives the SE for this case.
\begin{theorem}\label{Theorem:genstealthy}
	The attacker's strategy at the SE is: $ p^{a^*}_{gb} = \frac{\sum_{b \in \mathcal{B}_g}t^g_b p^t_b}{2|\mathcal{B}_g|}  $, where $ |\mathcal{B}_g| $ is the number of BSs which are connected to the power source $ g $ and the defender's strategy at the SE is the solution of the following linear program:
	\begin{align}\label{eq:linproggen}
		\max_{\boldsymbol{p}^d} &\,\,z_b p_b^d,\\
		\text{s.t.}& \,\,\sum_{b \in \mathcal{B}} p_b^d \leq P^d,\\
		 &\,\, 0 \leq p_b^d \leq \sum_{g \in \mathcal{G}} p^{a^*}_{gb}, \,\, \forall b \in \mathcal{B}.\label{eq:cond11}
	\end{align}
\end{theorem}

\begin{proof}
	Considering equal values for $ p^a_{gb} \geq p^a_g ,\forall b \in \mathcal{B}_g $ for every power source $ g $, we can rewrite \eqref{eq:maximizerpslevel} as:
	\begin{align}
		\max_{\boldsymbol{p}^a} \sum_{g \in \mathcal{G}}\left(p^a_{g}\sum_{b \in \mathcal{B}}z_b\right) \left(1- \frac{p^a_{g}|\mathcal{B}_g|}{\sum_{b \in \mathcal{B}}t^g_b p^t_b}\right) - \sum_{b \in \mathcal{B}} p_b^d.\label{eq:simplifiedgen}
	\end{align} 
	
We can solve this maximization problem by taking the partial derivative of \eqref{eq:simplifiedgen} with respect to $ p^a_g, \,\, \forall g \in \mathcal{G}$ and setting it to 0. Then we have:
\begin{equation*}
\frac{\partial  \sum_{g \in \mathcal{G}}\left(p^a_{g}\sum_{b \in \mathcal{B}}z_b\right) \left(1- \frac{p^a_{g}|\mathcal{B}_g|}{\sum_{b \in \mathcal{B}}t^g_b p^t_b}\right)}{\partial p^a_g} - \frac{ \partial \sum_{b \in \mathcal{B}} p_b^d}{\partial p^a_g} = 0.
\end{equation*}
which is equivalent to:
\begin{equation*}
\frac{\partial  \sum_{g \in \mathcal{G}}\left(p^a_{g}\sum_{b \in \mathcal{B}}z_b\right) \left(1- \frac{p^a_{g}|\mathcal{B}_g|}{\sum_{b \in \mathcal{B}}t^g_b p^t_b}\right)}{\partial p^a_g} = 0.
\end{equation*}
taking the derivatives of the individual terms:
\begin{equation*}
 \sum_{b \in \mathcal{B}}z_b \left(1- \frac{p^a_{g}|\mathcal{B}_g|}{\sum_{b \in \mathcal{B}}t^g_b p^t_b}\right) +  \left(p^a_{g}\sum_{b \in \mathcal{B}}z_b\right) \left(- \frac{|\mathcal{B}_g|}{\sum_{b \in \mathcal{B}}t^g_b p^t_b}\right)= 0.
\end{equation*}
then we have:
\begin{equation*}
\left(1- \frac{p^a_{g}|\mathcal{B}_g|}{\sum_{b \in \mathcal{B}}t^g_b p^t_b}\right) =  p^a_{g} \left( \frac{|\mathcal{B}_g|}{\sum_{b \in \mathcal{B}}t^g_b p^t_b}\right).
\end{equation*}
from which we can get the optimal strategy as:
\begin{equation*}
 p^{a^*}_{g} = \frac{\sum_{b \in \mathcal{B}_g}t^g_b p^t_b}{2|\mathcal{B}_g|}  
\end{equation*}

We note that this solution satisfies \eqref{cond10}, so it represents a valid solution to the problem in \eqref{eq:maximizerpslevel}.

Next, in order to find the defender's strategy at the SE, we first plug in the strategy of the attacker at the SE into \eqref{eq:defenderpayoff} as follows:
\begin{align}\label{eq:pluggedindef}
u^d(\boldsymbol{p}^d,\boldsymbol{p}^{a^*}) = -\sum_{b \in \mathcal{B}}z_b \left(\sum_{g \in \mathcal{G}}p^{a^*}_{gb} - p_b^d\right).
\end{align}

Since the first term in \eqref{eq:pluggedindef} does not depend on $ \boldsymbol{p}^d $, then the defender's problem simplifies to \eqref{eq:linproggen}. However, since the allocated BPS at every BS cannot exceed the attacker's deviation at the same BS, we add the constraints in (\ref{eq:cond11}) which must be satisfied by the defender.
\end{proof}

Theorem \ref{Theorem:genstealthy} shows that there exists only one SE for the case in which the attacker wants to stay stealthy at the power source level since the linear program in \eqref{eq:linproggen} will yield only one solution. To solve \eqref{eq:linproggen}, we can use known techniques such as the simplex method \cite{boyd2004convex}.

\subsection{Stealthiness at the Power Line Level}
In this case, the attacker's problem can be formulated according to \eqref{eq:attpayoffline} as follows:
\begin{align}\label{eq:maxlin}
	&\max_{\boldsymbol{p}^a} \sum_{b \in \mathcal{B}}z_b \left(\sum_{g \in \mathcal{G}}p^a_{gb} \left(1-\frac{p^a_{gb}}{t^g_b p^t_b}\right) - p_b^d\right)\\
	&\text{s.t.} \,\, p^a_{gb} \leq t^g_b p^t_b, \forall g \in \mathcal{G},\& \,\,  b \in \mathcal{B}_g.
\end{align}

Next, we derive the optimal strategy of the attacker that maximizes its payoff, plug it into the defender's strategy, and, then, derive the maximizer strategy of the defender, similar to the case of power source level.

\begin{theorem}\label{Theorem:LineStealthy}
At the power line level, the SE of the game occurs when the attacker plays the strategy is: $ p_{gb}^{a^*} = \frac{t_b^gp_b^t}{2} $ and the defender plays the strategy given by the solution of the following linear program:
	\begin{align}\label{eq:linproglin}
		\max_{\boldsymbol{p}^d} &\,\,z_b p_b^d,\\
		\text{s.t.}& \,\,\sum_{b \in \mathcal{B}} p_b^d \leq P^d,\\
		&\,\, 0 \leq p_b^d \leq \sum_{g \in \mathcal{G}} \frac{t_b^gp_b^t}{2}, \,\, \forall b \in \mathcal{B}.\label{eq:cond21}
	\end{align}
\end{theorem}

\begin{proof}
	The proof follows a similar procedure to the proof of Theorem \ref{Theorem:genstealthy}, where we first take the partial derivative of \eqref{eq:maxlin} with respect to $ p_{gb}^a$ and setting it equal to zero. Then we can find the maximizer strategy of the attacker which is $ p_{gb}^{a^*} = \frac{t_b^gp_b^t}{2} $.
	
	By plugging this value into the payoff function of the defender in \eqref{eq:defenderpayoff}, we end up having the linear program in \eqref{eq:linproglin} and the condition in \eqref{eq:cond21} comes from the fact that the allocated BPSs at every BS cannot exceed the power deviation caused by the attacker at each BS.
\end{proof}

Theorem \ref{Theorem:LineStealthy} shows that, our game will admit a unique SE for the case in which the attacker tries to stay stealthy at the power line level. 

\subsection{Stealthiness at the BS Level}
In this case, the attacker's problem can be formulated according to \eqref{eq:attpayoffbs} as follows:
\begin{align}\label{eq:maximizerbs}
	&\max_{\boldsymbol{p}^a} \sum_{b \in \mathcal{B}}z_b \left(\sum_{g \in \mathcal{G}}p^a_{gb} - p_b^d\right)\left(1-\frac{\sum_{g \in \mathcal{G}}p^a_{gb}}{p^t_b - p^o_b}\right)\\
	& \text{s.t.}\,\, \sum_{g \in \mathcal{G}}p^a_{gb}\leq p^t_b - p^o_b, \,\, \forall b \in \mathcal{B}.
\end{align}
The following theorem derives the SE for this case.

\begin{theorem}\label{Theorem:BSStealthy}
	At the SE, the defender's strategy is the solution of the following linear program:
	\begin{align}\label{eq:linprogbs}
	\max_{\boldsymbol{p}^d} &\,\,z_b p_b^d,\\
	\text{s.t.}& \,\,\sum_{b \in \mathcal{B}} p_b^d \leq P^d,\\
	&\,\, 0 \leq p_b^d \leq \sum_{g \in \mathcal{G}} \frac{p_b^t-p_b^o}{2}, \,\, \forall b \in \mathcal{B},\label{eq:cond31}
	\end{align}
	while the attacker's SE strategy is any $ \boldsymbol{p}^{a^*} $, such that $ \sum_{g \in \mathcal{G}}p^a_{gb} = \frac{p_b^{d^*} + p_b^t-p_b^o}{2}, \,\, \forall b \in \mathcal{B} $.
\end{theorem}

\begin{proof}
	In order to find the maximizer strategy of the attacker, we define $ p_b^a \triangleq \sum_{g \in \mathcal{G}}p^a_{gb} $. In this case, the attacker's problem simplifies to:
	\begin{align}\label{eq:maximizerbssimple}
		&\max_{\boldsymbol{p}^a} \sum_{b \in \mathcal{B}}z_b \left(p^a_{b} - p_b^d\right)\left(1-\frac{p^a_{b}}{p^t_b - p^o_b}\right)\\
		& \text{s.t.}\,\, p^a_{b}\leq p^t_b - p^o_b, \,\, \forall b \in \mathcal{B}.
	\end{align}
	
	We can now take the first order partial derivative of \eqref{eq:maximizerbssimple} with respect to $ p_b^a  $ and set it equal to 0, similar to Theorem \ref{Theorem:genstealthy}. Then, the maximizer value for the attacker is derived as $ \frac{p_b^{d} + p_b^t-p_b^o}{2}, \,\, \forall b \in \mathcal{B} $.
	
	By substituting this value into the defender's payoff function, we end up having the linear program in \eqref{eq:linprogbs}. However, here, we note that $ p^a_{b}\leq p^t_b - p^o_b $ should be satisfied which yields to the condition in \eqref{eq:cond31}. This linear program has a unique solution, however, since any solution of $  p_b^a \triangleq \sum_{g \in \mathcal{G}}p^{a^*}_{gb} $ is an SE, thus, there exists infinitely many SEs for this case.
\end{proof}

Theorem \ref{Theorem:BSStealthy} shows that, for the case in which the attacker tries to stay stealthy at the BS level, there are infinitely many SEs for the game. However, all of these strategies will yield the same payoff for the defender and the attacker.

Finally, we can see that all the derived SEs are the solutions of linear programs. Therefore, the proposed game model and its solution(s) can be solved in linear time with respect to the size of the ICI. This property makes the proposed solution scalable and, hence, it is applicable to large-scale ICIs.

\section{Simulation Results and Analysis}\label{sec:sim}
For our simulations, we consider a grid model for the ITS in which the lengths of streets and the numbers of intersections in the x and y directions are equal. We consider hexagonal models for the communication network's cells and we spread a number of power sources in the studied world where their locations are drawn according to a 2-D uniform distribution. The number of BSs that are connected to the power sources is also drawn according to a uniform distribution. For every simulation, we construct a new flow matrix $ \boldsymbol{A} $ that satisfies the condition in \eqref{eq:intersection}. In addition, we consider that for all of the BSs, $ p_b^t = 2p_b^o $.

First, in Fig. \ref{fig:powdev}, we show the impact of deviation of the power generation of all power sources on the flow deviation of the ITS. We can see that, as the percentage of power reduction increases, the flow deviation linearly increases until when the power is reduced to half of the total power generation. At this point $ p_b^r = p_b^o $, which means that the received power can only activate the BS but is not enough to send control packets to ACVs. That is why after $50\%$ reduction the flow deviation stays constant.
\begin{figure}[t]
	\centering
	\includegraphics[width=\columnwidth]{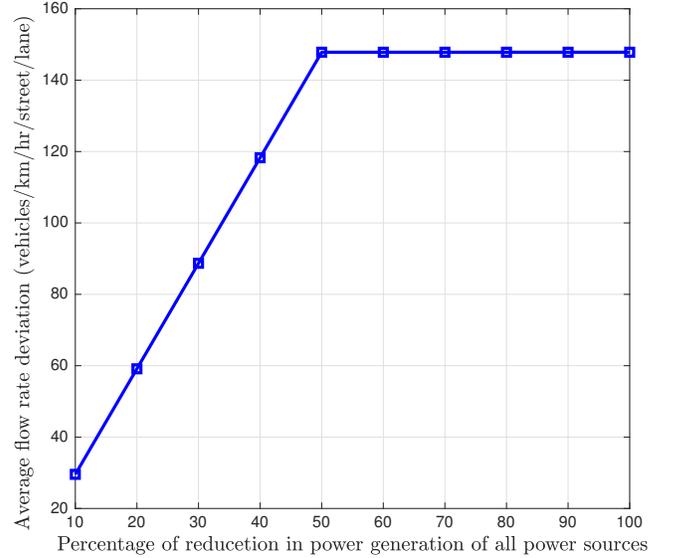}
	\caption{Effect of the power deviation on the ITS flow.}
	\label{fig:powdev}
\end{figure}

\begin{figure}
	\centering
	\begin{subfigure}{\columnwidth}
		\centering
		\includegraphics[width=\columnwidth]{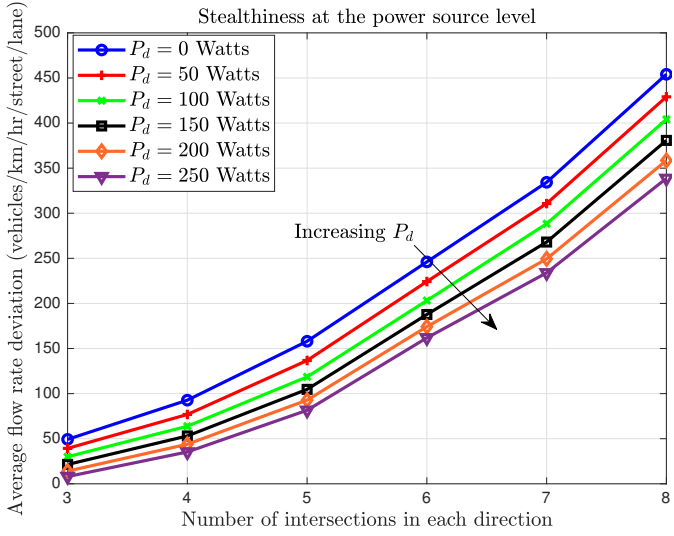}
		\caption{Stealthiness at the power source level.}
		\label{fig:powersource}
	\end{subfigure}\\
	\begin{subfigure}{\columnwidth}
		\centering
		\includegraphics[width=\columnwidth]{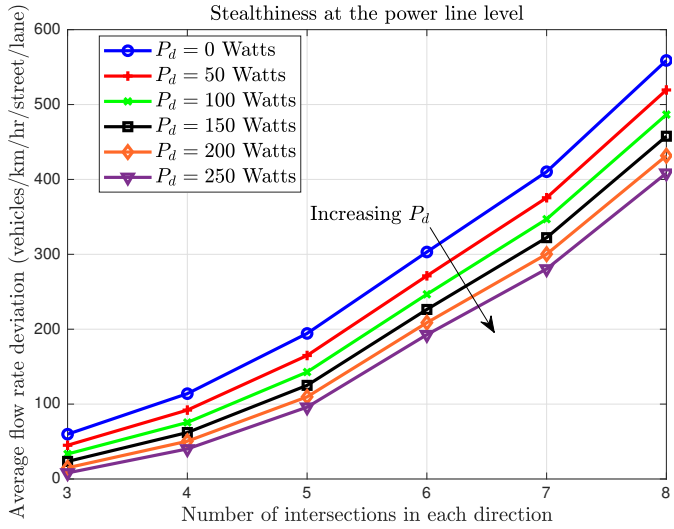}
		\caption{Stealthiness at the power line level.}
		\label{fig:powerline}
	\end{subfigure}\\
	\begin{subfigure}{\columnwidth}
		\centering
		\includegraphics[width=\columnwidth]{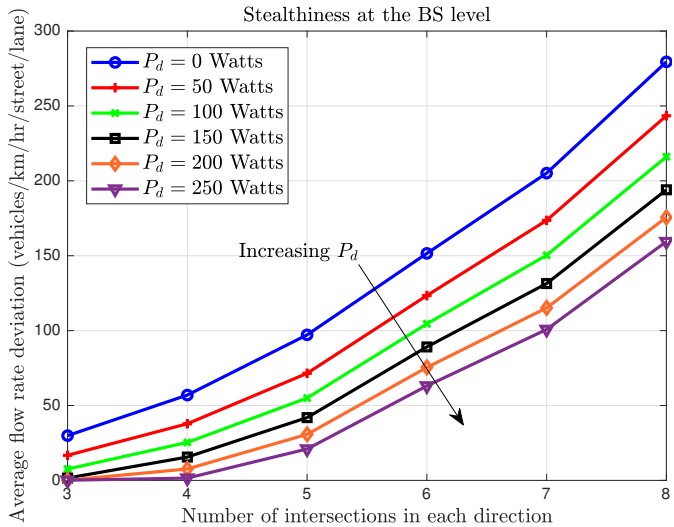}
		\caption{Stealthiness at the BS level.}
		\label{fig:bs}
	\end{subfigure}
	\caption{Effect of the ITS scale on the traffic flow deviation for different levels of stealthiness.}
	\label{fig:scale}
\end{figure}

Fig. \ref{fig:scale} illustrates the effect of the scale of the ITS on the SE strategies and the flow deviation. In this scenario, we keep both the number of power sources and the radius of cells constant. However, the number of the cells will increase as the scale of the ITS grows.
In Fig. \ref{fig:scale}, we also compare the three types of stealthiness level that are power source level (Fig. \ref{fig:powersource}), power line level (Fig. \ref{fig:powerline}), and BS level (Fig. \ref{fig:bs}). First, we can see that in all the cases, the increase in the scale of the ITS, represented by the number of the intersections in each direction, will increase the interdependence of the ITS on the CI thus causing higher flow deviations on the ITS. The upper line in each figure represents the case of zero $P_d$, i.e., no BPSs are used, and hence, the highest deviation. As the values of the compensated power increase, we can see a reduction in the flow deviation, in all the cases.

From Fig. \ref{fig:scale}, we can also see that on average, power line level stealthiness, in Fig. \ref{fig:powerline}, can cause higher flow deviation than power source level stealthiness, Fig. \ref{fig:powersource}, and, than, BS level stealthiness, Fig. \ref{fig:bs}. That means the interdependent PG-CI-ITS infrastructure is more sensitive to power line attacks than the other types of attacks. We can also in Fig. \ref{fig:scale}, that compensating the power by using BPSs at the BS level has higher positive effect on reducing the deviations in the flow. This can be seen from Fig. \ref{fig:bs} as the gap between the lines representing flow deviations due to additional $P_d$ are wider than that in Figs. \ref{fig:powersource} and \ref{fig:powerline}. Therefore, at the value of $P_d = 250 $ watts, for instance, the flow deviation is the minimum for all number of intersections in Fig. \ref{fig:bs}.

Fig. \ref{fig:radius} shows the effect of the cell coverage radius on the final flow deviation. From Fig. \ref{fig:radius}, we can see that, in general, an increase in the cell size reduces the flow rate deviation. However, for the case in which the cell radius is 1.2 km there is a small increase in the average flow rate deviation, for most of the cases, compared to when the cell radius is 1.1 km.
This is because, for this value, the increase in the radius means that fewer BSs are needed to cover the ITS. This, in turn, increases the portion of the ITS that is covered by each BS allowing the failures of the BSs to cause a higher deviation in the average flow rate of the ITS. However, this increase is limited due to the fact that each BS becomes connected to more generators which makes the BSs less vulnerable to power deviations. Note that, in Fig. \ref{fig:radius}, increasing the cell size over 1.2 km does not reduce the number of BSs, and thus, the flow rate deviation returned to normal behavior as before the value of 1.2 km.

\begin{figure}[t]
	\centering
	\includegraphics[width=\columnwidth]{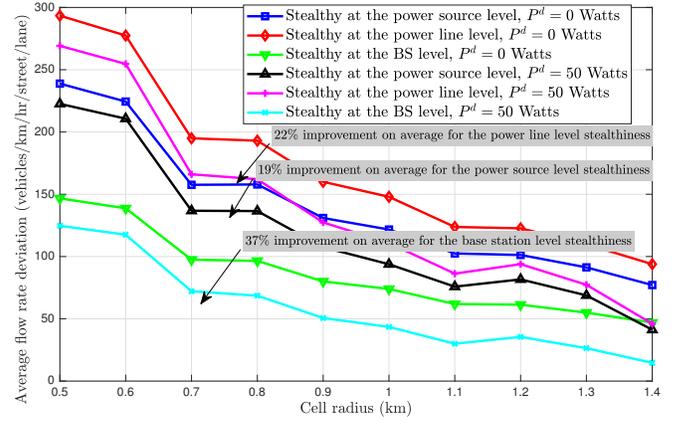}
	\caption{Effect of the cell radius on the power flow deviation.}
	\label{fig:radius}
\end{figure}

Moreover, Fig. \ref{fig:radius} shows that our proposed Stackelberg allocation strategy can reduce the flow deviation by 19\%, 22\%, and 37\% for stealthiness levels at power source, power line, and BS respectively, when the value of $P_d=50$ watts. This is represented by the gap between the lines representing the same stealthiness level when the BPSs are allocated per the optimal SE equilibrium in Section \ref{sec:game}.

\begin{figure}[t]
	\centering
	\includegraphics[width=\columnwidth]{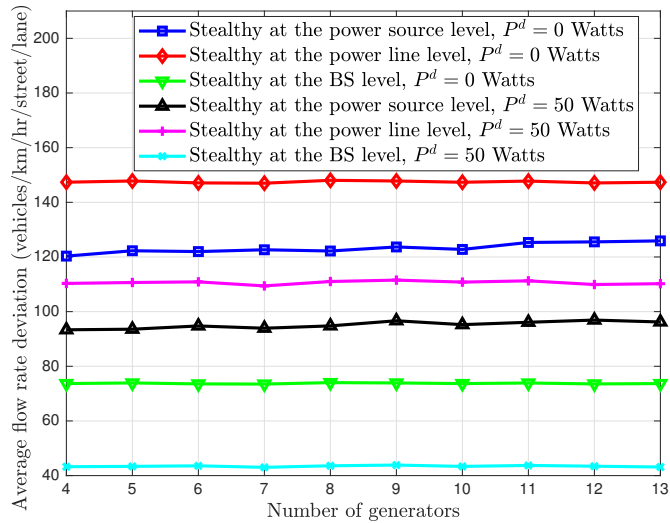}
	\caption{Attack on all of the power sources.}
	\label{fig:num_gen_all}
\end{figure}

\begin{figure}[t]
	\centering
	\includegraphics[width=\columnwidth]{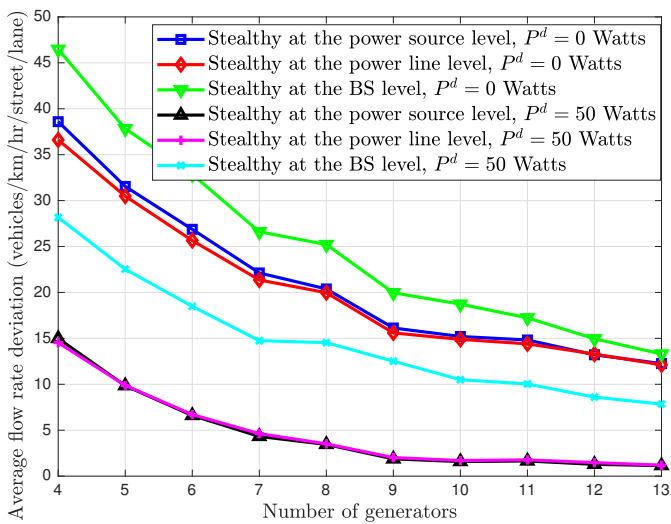}
	\caption{Attack on only one power source.}
	\label{fig:num_gen_single}
\end{figure}

Figs. \ref{fig:num_gen_all} and \ref{fig:num_gen_single} show the effect of the number of power sources on the flow deviation. In Fig. \ref{fig:num_gen_all}, we study the case in which the attacker performs an attack on all of the power sources, while in Fig. \ref{fig:num_gen_single}, we show the case in which the attacker performs its attack on only one of the power sources. From Fig. \ref{fig:num_gen_all}, we observe that the flow deviation does not change as the number of generators increase. This due to the fact that, by adding more power sources, the attacker can still attack the newly added power sources and, thus, the new power sources would not compensate the power loss at the BSs. Note that, this case might be unrealistic in the sense that even powerful attackers might not have the capacity to attack all the power sources in the PG. However, we show this case to highlight the difference when the attacker attack just few generators. In particular, we consider the case of attacking just one generator in Fig. \ref{fig:num_gen_single}. From Fig. \ref{fig:num_gen_single} we can see that, when the number of power sources increases, the flow deviation becomes smaller, when the number of the generators increase. This is because the ITS can now benefit from the added power sources as they are not targeted by the attacker.

Finally, in Fig. \ref{fig:comp}, we compare our proposed SE allocation with an allocation strategy in which the BPSs are equally distributed at every BS. We can see from Fig. \ref{fig:comp} that our proposed SE can reduce the flow deviation by up to 40\% compared to the case of equal allocation. Moreover, when the attacker is not stealthy, the equal BPS distribution cannot defend the ITS against the attacker compared to the SE allocation which is shown to be effective in reducing the flow deviation. The other stealthiness levels exhibit similar behavior where the SE allocations outperform the equal allocations in all cases. We also can see from Fig. \ref{fig:comp} that when the attacker's goal is to stay stealthy at the power lines, it can cause higher flow deviations compared to the other two stealthiness levels, which corroborates the results in Fig. \ref{fig:scale}.

\begin{figure}[t]
	\centering
	\includegraphics[width=\columnwidth]{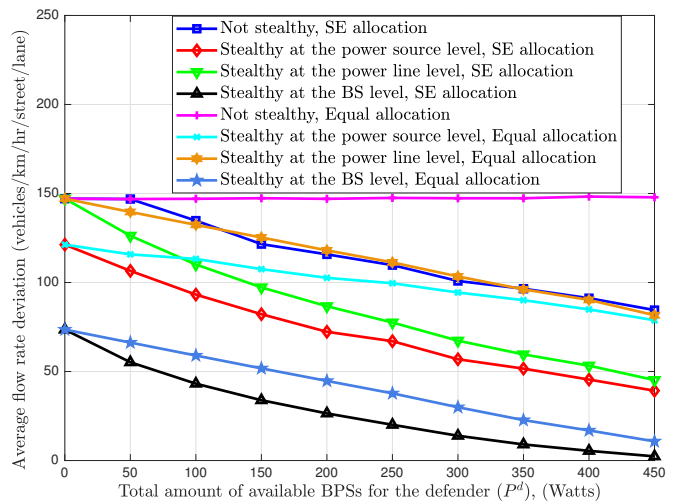}
	\caption{Comparison between the SE and equal allocations.}
	\label{fig:comp}
\end{figure}

\section{Conclusion}\label{sec:conc}

In this paper, we have studied the security of interdependent PG-CI-ITS infrastructure. We have modeled each infrastructure where it was shown that there is a strong interdependence between the ITS, the CI, and the PG infrastructure. Using these individual models, we have derived a rigorous one-to-one interdependence relation that can map the effect of power loss at the PG components on the ITS traffic flow. Then, we have defined the possible ways in which an attacker can perform a stealthy attack on the interdependent infrastructure. In particular, three levels of attacker stealthiness have been considered that can occur at the power source, the power lines, and at the BS level. In order to defend against any these types of attacks, we have proposed a Stackelberg game to model the interactions between the attacker and the infrastructure administrator, as the defender. We have analytically derived the Stackelberg solution for the different levels of stealthiness for the attacker. This Stackelberg solution can be used by the defender to strategically allocate its available BPSs at every BS. 
We have also shown that the solutions of the proposed game are scalable as they can be reached in linear time with respect to the size of ICI which makes the analysis practical for large-scale ICIs. Results have shown that the proposed Stackelberg allocation outperforms other strategy selection techniques and, in particular, can reduce the flow deviation at ITS up to 40\% compared to an equal BPS allocation strategy.

\def\baselinestretch{1}
\bibliographystyle{IEEEtran}
\bibliography{references}

\end{document}